\documentclass[a4paper,11pt]{article}
\title{\bf Socially stable matchings in the \\Hospitals/Residents problem}
\author{Georgios Askalidis$^1$, Nicole Immorlica$^{1,2}$,  Augustine Kwanashie$^3$, \\
David F. Manlove$^4$ and  Emmanouil Pountourakis$^5$\\ \\
{\small \emph{$^{1,2,5}$Department of Electrical Engineering and Computer Science, Northwestern University, USA}}\\
{\small \emph{$^2$Microsoft Research New England, USA}}\\
{\small \emph{$^{3,4}$School of Computing Science, University of Glasgow, UK}} \\
{\small \tt{$^{1}$gask@u.northwestern.edu}, \tt{$^{2}$nicimm@gmail.com},  \tt{$^{3}$a.kwanashie.1@research.gla.ac.uk}}
\\{\small \tt{$^{4}$David.Manlove@glasgow.ac.uk}, \tt{$^{5}$manolis@u.northwestern.edu}}}

\usepackage{makeidx}
\usepackage{hyperref}
\usepackage{pdflscape}
\usepackage{amsmath}
\usepackage{amsfonts}
\usepackage{verbatim}
\usepackage{graphicx}
\usepackage{algorithm}
\usepackage{algorithmic}
\usepackage{textcomp}
\usepackage{mathtools}
\usepackage{float}
\usepackage{epstopdf}
\usepackage{tikz}
\usetikzlibrary{arrows}
\usetikzlibrary{automata}
\usetikzlibrary{positioning} 

\newtheorem{lemma1}{Lemma}[section]
\newtheorem{theorem1}[lemma1]{Theorem}
\newtheorem{proposition1}[lemma1]{Proposition}
\newtheorem{definition1}[lemma1]{Definition}

\newtheorem{claim}[lemma1]{Claim}
\newenvironment{proof}[1][Proof]{\begin{trivlist}
\item[\hskip \labelsep {\bfseries #1}]}{\end{trivlist}}
\newcommand{\qed}{\hfill \ensuremath{\Box}}

\begin{document}
\setcounter{page}{1}
\bibliographystyle{plain}
\maketitle
\pagenumbering{arabic}

\begin{abstract}
In the Hospitals/Residents (HR) problem, agents are partitioned into hospitals and residents.  Each agent wishes to be matched to an agent (or agents) in the other set and has a strict preference over these potential matches.  A matching is stable if there are no blocking pairs, i.e., no pair of agents that prefer each other to their assigned matches. Such a situation is undesirable as it could lead to a deviation in which the blocking pair form a private arrangement outside the matching. This however assumes that the blocking pair have social ties or communication channels to facilitate the deviation. Relaxing the stability definition to take account of the potential lack of social ties between agents can yield larger stable matchings.

In this paper, we define the Hospitals/Residents problem under Social Stability (HRSS) which takes into account social ties between agents by introducing a \emph{social network graph} to the HR problem.  Edges in the social network graph correspond to resident-hospital pairs in the HR instance that know one another. Pairs that do not have corresponding edges in the social network graph can belong to a matching $M$ but they can never block $M$. Relative to a relaxed stability definition for HRSS, called \emph{social stability}, we show that socially stable matchings can have different sizes and the problem of finding a maximum socially stable matching is NP-hard, though approximable within $3/2$. Furthermore we give polynomial time algorithms for special cases of the problem.
\end{abstract}

\section{Introduction}
Matching problems generally involve the assignment of a set (or sets) of agents to one another. Agents may be required to list other agents they find acceptable in order of preference, either explicitly or implicitly through a list of desirable characteristics. Agents may also be subject to capacity constraints, indicating the maximum number of assignments they are allowed to be involved in.

An example of such a matching problem that has received much attention in literature is the \emph{Hospitals/Residents problem (HR)} \cite{GS62,GI89,RS90,Man13}. An HR instance consists of a set of \emph{residents} seeking to be matched to a set of \emph{hospitals}.  Each resident ranks a subset of the hospitals in strict order of preference, and vice versa.  Further, each resident forms an \emph{acceptable pair} with every hospital on his preference list.  Finally, each hospital has a \emph{capacity}, indicating the maximum number of residents that it can be assigned.  A \emph{matching} is a set assignments among acceptable pairs such that no resident is assigned to more than one hospital, and no hospital exceeds its capacity.  An acceptable pair forms a \emph{blocking pair} with respect to a matching, or \emph{blocks} a matching, if both agents would rather be assigned to each other than remain with their assignees (if any) in the matching. A matching is \emph{stable} if it admits no blocking pair. 
%
%
HR has a wide range of applications including traditional markets like the assignment of graduating medical students (residents) to hospitals \cite{Irv98,Rot84} and students to high schools \cite{APR05a,APR05}, and online markets like oDesk (an online labour market), AirBnB (an online short-term housing rental market), and Match.com/OkCupid/etc.\ (online dating markets). In applications such as these, it has been convincingly argued that stability is a desirable property of a matching \cite{Rot84}.

Although the concept of stability is important in many applications of matching problems, there are classes of matching problems (such as the Stable Roommates problem) for which an instance is not guaranteed to admit a stable matching \cite{GS62}. Moreover, enforcing the stability requirement tends to reduce the size of matchings discovered \cite{BMM10}. This is an issue particularly in the case of applications where it is desirable to find the largest possible matching. Also, it is generally assumed that a resident-hospital pair that blocks a matching in theory will also block the matching in practice. However this assumption is not always true in some real-life applications, as resident-hospital pairs are more likely to form blocking pairs in practice if social ties exist between them. These factors have motivated studies into alternative, weaker stability definitions that still aim to prevent a given matching from being subverted in practice while increasing the number of agents involved in the matchings.

Arcaute and Vassilvitskii \cite{ES09b} described the Hospitals/Residents problem in the context of assigning job applicants to company positions. They observed that applicants are more likely to be employed by a company if they are recommended by their friends who are already employees of that company. In their model, an applicant-company pair $(a, c)$ may block a matching $M$ if $(a, c)$ blocks it in the traditional sense (as described in the analogous HR context) and $a$ \emph{is friends with} another applicant $a'$ assigned to $c$ in $M$. Thus their problem incorporates both the traditional HR problem and additionally an underlying social network, represented as an undirected graph consisting of applicants as nodes and edges between nodes where the corresponding applicants have some social ties (e.g., are friends). Matchings that admit no blocking pair in this context are called \emph{locally stable} due to the addition of the informational constraint on blocking pairs. Cheng and McDermid \cite{CM12} investigated  the problem (which they called HR+SN) further and established various algorithmic properties and complexity results. They showed that locally stable matchings can be of different sizes and the problem of finding a maximum locally stable matching is NP-hard. They identified special cases where the problem is polynomially solvable and gave upper and lower bounds on the approximability of the problem.

While the HR+SN model is quite natural in the job market, it makes an assumption that the employed applicant $a'$ will always be willing to make a recommendation. This however may not be the case as a recommendation may in practice lead to $a'$ being rejected by his assigned company. Ultimately this may lead to a reassignment for $a'$ to a worse company or indeed $a'$ may end up unmatched. While it is true that a scenario may arise where these social ties between applicants may lead to a blocking pair of a matching, it is arguably equally likely that social ties between an applicant and the company itself will exist. That is, an applicant need not know another applicant who was employed by the company in order  to block a matching; it is enough for him to know \emph{any} employee in the company (for example the Head of Human Resources). Such a model could also be natural in many applications both within and beyond the job market context. 

Additionally, many matching markets are cleared by a centralised clearinghouse.  While more traditional markets require agents to explicitly list potential matches, many online markets ask agents to list desirable characteristics and then use software to infer the preference lists of the agents.  In these markets, communication between agents is facilitated by the centralised clearinghouse.  Some agent pairs in the market may have social ties outside the clearinghouse.  Often these social ties are due to past interactions within the marketplace and so the clearinghouse is aware of them.  These pairs can communicate outside the clearinghouse and might block proposed matchings.  Most pairs, however, only become aware of each other when the clearinghouse proposes them as a match.  Thus even if they prefer each other to their assigned matches, they will not be able to discover each other and deviate from the matching.

Based on these ideas, we present a variant of HR called the \emph{Hospitals / Residents problem under Social Stability (HRSS)}.  In this model, which we describe in the context of assigning graduating medical residents to hospital positions, we assume that a resident-hospital pair will only form a blocking pair in practice if there exists some social relationship between them. Two agents that have such a social relationship are called an \emph{acquainted pair}, and this is represented by an edge in a \emph{social network graph}. We call a pair of agents that do not have such a social relationship an \emph{unacquainted pair}.  Such a pair may be part of a matching $M$ (given that $M$ is typically constructed by a trusted third party, i.e., a centralised clearinghouse) but cannot form a blocking pair with respect to $M$.  As a consequence, although a resident-hospital pair may form a blocking pair in the classical sense, if they are an unacquainted pair, they will not form a blocking pair in the HRSS context. A matching that admits no blocking pair in this new context is said to be \emph{socially stable}. We denote the one-to-one restriction of HRSS as the \emph{Stable Marriage problem with Incomplete lists under Social Stability (SMISS)}.

Hoefer and Wagner \cite{Hoe13,HW13} studied a problem that generalises both HR+SN and HRSS.  In their model, the social network graph involves all agents and need not be bipartite.  A pair \emph{locally blocks} a given matching $M$ if (i) it blocks in the classical sense, and (ii) the agents involved are at most $l$ edges apart in the social network graph augmented by $M$.  This scenario can be viewed as a generalisation of the HR+SN ($l = 2$) and HRSS ($l=1$) models.  They studied the convergence time for better-response dynamics that converge to locally stable matchings, and also established a lower bound for the approximabiliy of the problem of finding a maximum locally stable matching (for the case that $l \leq 2$).

Locally stable matchings have also been investigated in the context of the Stable Roommates problem (a non-bipartite generalisation of the Stable Marriage problem) in \cite{CF09}. Here, the \emph{Stable Roommates problem with Free edges (SRF)} as introduced was motivated by the observation that, in kidney exchange matching schemes, donors and recipients do not always have full information about others and are more likely to have information only on others in the same transplant centre as them. The problem is defined by the traditional Stable Roommates problem together with a set of \emph{free edges}.  These correspond to pairs of agents in different transplant centres that do not share preference information; such pairs may be involved in stable matchings, but cannot block any matching. It is shown in \cite{CF09} that the problem of determining whether a stable matching exists, given an SRF instance, is NP-complete.

In this paper, we present some algorithmic results for the HRSS model described above. In Section \ref{pre}, we present some preliminary definitions and observations. We give a reduction from the HRSS to the HR+SN problem in Section \ref{reduction}. In Section \ref{hrf} we show that MAX HRSS, the problem of finding a maximum socially stable matching given an HRSS instance $I$, is NP-hard even under certain restrictions on the lengths of the preference lists. This result holds even if $I$ is an instance of SMISS. Then in Section \ref{approx_smif}, we consider the approximability of MAX HRSS.  We give a $3/2$-approximation algorithm for the problem, and also show that it is not approximable within $21/19-\varepsilon$, for any $\varepsilon >0$, unless P=NP, and not approximable within $3/2-\varepsilon$, for any $\varepsilon > 0$, assuming the Unique Games Conjecture.  in Section \ref{special} we present polynomial-time algorithms for three special cases of MAX HRSS, including the cases where (i) each resident's list is of length at most 2 and each hospital has capacity 1, (ii) the number of unacquainted pairs is constant, and (iii) the number of acquainted pairs is constant.  Finally some open problems are given in Section \ref{open}.

\section{Preliminary definitions and results}
\label{pre}
An instance $I$ of the Hospitals/Residents problem (HR), as defined in \cite{GS62}, contains a set $R=\{r_1, r_2, ... , r_{n_1}\}$ of residents, a set $H=\{h_1, h_2, ... , h_{n_2}\}$ of hospitals. Each resident $r_i \in R$ ranks a subset of $H$ in strict order of preference; each hospital $h_j \in H$ ranks a subset of $R$, consisting of those residents who ranked $h_j$, in strict order of preference. Each hospital $h_j$ also has a capacity $c_j \in \mathbb{Z}^+$  indicating the maximum number of residents that can be assigned to it. A pair $(r_i, h_j)$ is called an \emph{acceptable pair} if $h_j$ appears in $r_i$'s preference list. We denote by $\mathcal{A}$ the set of all acceptable pairs. A \emph{matching} $M$ is a set of acceptable pairs such that each resident is assigned to at most one hospital and the number of residents assigned to each hospital does not exceed its capacity. If $r_i$ is matched in $M$, we denote the hospital assigned to resident $r_i$ in $M$ by $M(r_i)$. We denote the set of residents assigned to hospital $h_j$ in $M$ as $M(h_j)$. A resident $r_i$ is \emph{unmatched} in $M$ if no pair in $M$ contains $r_i$.  A hospital $h_j$ is \emph{undersubscribed} in $M$ if $|M(h_j)| < c_j$. A pair $(r_i, h_j)$ is said to \emph{block} a matching $M$, or form a \emph{blocking pair} with respect to $M$, in the classical sense, if (i) $r_i$ is unmatched in $M$ or prefers $h_j$ to $M(r_i)$ and (ii) $h_j$ is undersubscribed in $M$ or prefers $r_i$ to some resident in $M(h_j)$. A matching that admits no blocking pair is said to be \emph{stable}.

We define an instance $(I, G)$ of the \emph{Hospitals/Residents Problem under Social Stability (HRSS)} as consisting of an HR instance $I$ (as defined above) and a bipartite graph $G = (R \cup H, A)$, where $A\subseteq \mathcal A$. A pair $(r_i, h_j)$ belongs to $A$ if and only if $r_i$ has social ties with $h_j$. We call $(r_i, h_j)$ an \emph{acquainted pair}.  We also define the set of \emph{unacquainted} pairs (which cannot block any matching) to be $U = \mathcal{A} \backslash A$. A pair $(r_i, h_j)$ \emph{socially blocks} a matching $M$, or forms a \emph{social blocking pair} with respect to $M$, if  $(r_i, h_j)$ blocks $M$ in the classical sense in the underlying HR instance $I$ and $(r_i, h_j) \in A$.  A matching $M$ is said to be \emph{socially stable} if there exists no social blocking pair with respect to $M$. If we restrict the hospitals' capacities to 1, we obtain the \emph{Stable Marriage problem with Incomplete lists under Social Stability (SMISS)}, and refer to the agents as \emph{men} $\mathcal{U} =\{m_1,\ldots,m_{n_1}\}$ and \emph{women} $\mathcal{W}=\{w_1,\ldots,w_{n_2}\}$.

Clearly every instance of HRSS admits a socially stable matching.  This is because the underlying HR instance is bound to admit a stable matching \cite{GS62} which is also socially stable. However socially stable matchings could be larger than stable matchings. Consider the SMISS instance $(I, G)$ shown in Figure \ref{smiss_instance}. Matchings $M_1 = \{(m_1, w_1), (m_2, w_2)\}$ and $M_2 = \{(m_2, w_1)\}$ are both socially stable in $(I, G)$ and $M_2$ is the unique stable matching. Thus an instance of SMISS (and hence HRSS) can admit a socially stable matching that is twice the size of a stable matching. Clearly the instance shown in Figure \ref{smiss_instance} can be replicated to give an arbitrarily large SMISS instance with a socially stable matching that is twice the size of a stable matching. This, and applications where we seek to match as many agents as possible, motivates MAX HRSS.

\begin{figure}[t]
\centering
\begin{minipage}[b]{0.6\linewidth}
\begin{align*}
&\mbox{men's preferences}  ~~~~~~~ &&\mbox{women's preferences}\\
&m_1: w_1  ~~~~~~~ &&w_1: m_2~~~m_1 \\
&m_2: w_1~~~w_2~~~~~~~ &&w_2:m_2 
\end{align*}
\end{minipage}
\begin{minipage}[b]{0.3\linewidth}
\centering
social network graph $G$\\
\begin{tikzpicture}
      \tikzstyle{every node}=[draw,circle,fill=black,minimum size=3pt,
                            inner sep=0.3pt]
                                                  
\draw (0,0) node (m1) [label=left:$m_2$]{};
\draw (2,0) node (w1) [label=right:$w_2$]{};
\path (m1) edge (w1);

\draw (0,1) node (m2) [label=left:$m_1$]{};
\draw (2,1) node (w2) [label=right:$w_1$]{};
\path (m2) edge (w2);
\end{tikzpicture}
\end{minipage}
\caption{SMISS instance $(I, G)$}
\label{smiss_instance}
\end{figure}
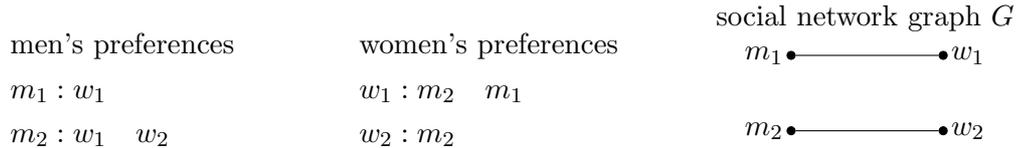
\section{Reduction from HRSS to HR+SN}
\label{reduction}
As defined in \cite{ES09b,CM12}, an instance $(I, G')$ of the HR+SN problem involves a Hospitals / Residents instance $I$, defined in \cite{GS62}, containing a set $R=\{r_1, r_2, ... , r_{n_1}\}$ of residents, a set $H=\{h_1, h_2, ... , h_{n_2}\}$ of hospitals, and a graph $G$ describing the social network (SN) of the residents. In the graph $G' = (V, E)$, $V = R$ and an edge $\{r_i, r_k\}$ belongs to $E$ if and only if $r_i$ and $r_k$ have \emph{social ties}. A pair $(r_i, h_j)$ is a \emph{local blocking pair} with respect to a matching $M$, or \emph{locally blocks} $M$, if $(r_i, h_j)$ blocks $M$ in the classical sense and there is some resident $r_k$ such that $\{r_i, r_k\} \in E$ and $r_k \in M(h_j)$. A matching $M$ is said to be \emph{locally stable} if there exists no local blocking pair with respect to $M$. In the HR+SN (respectively HRSS) context we refer to a \emph{resident-complete} locally (respectively socially) stable matching as one in which all the residents are matched.

In this section we show the close relationship between the HRSS and HR+SN problems. Consider an instance $(I, G)$ of HRSS where $I$ is the underlying HR instance and $G$ is the social network graph. $I$ involves of a set of residents $R_0 = \{r_1, r_2, ..., r_{n_1}\}$ and a set of hospitals $H_0 = \{h_1, h_2, ..., h_{n_2}\}$. We construct an instance $(I',G')$ of HR+SN from $(I,G)$ as follows: let $I'$ consist of a set of residents $R = R_0 \cup R_1$ where $R_1 = \{r_{n_1+1}, r_{n_1+2}, ..., r_{n_1+n_2}\}$. Every resident $r_{n_1+j} \in R_1$  has a single entry $h_j$ in his preference list. Every resident $r_i \in R_0$ has an identical preference list in $I'$ as in $I$.  Let $I'$ also involve a set of hospitals $H$, where $H = H_0$ such that every hospital $h_j \in H$ has resident $r_{n_1+j}$ as the first entry in its preference list and has capacity $c_j' = c_j+1$. $h_j$'s preference list in $I$ is then appended to $r_{n_i+j}$ to yield $h_j$'s preference list in $I'$. To construct $G'$, let the vertices in $G'$ correspond to the residents in $R$ and add edge $\{r_i, r_{n_1+j}\}$ to $G'$ if and only if $(r_i, h_j) \in A$, where $A=E(G)$ is the set of acquainted pairs in $(I,G)$.

\begin{theorem1} 
\label{hrsn_reduction}
If $M$ is a socially stable matching in $(I, G)$, then $M' = M \cup \{(r_{n_1+j}, h_j) : r_{n_1+j} \in R_1\}$ is a locally stable matching in $(I', G')$. Conversely if $M'$ is a resident-complete locally stable matching in $(I', G')$ then $M = M' \backslash \{(r_{n_1+j}, h_j) : r_{n_1+j} \in R_1\}$ is a resident-complete socially stable matching in $(I, G)$.
\end{theorem1}

\begin{proof}
Suppose $M$ is socially stable in $(I, G)$. Then no (classical) blocking pair with respect to $M$ in $I$ is contained in $G$.  Let $M' = M \cup \{(r_{n_1+j}, h_j) : r_{n_1+j} \in R_1\}$. If some pair ($r, h$) locally blocks $M'$ in $(I', G')$ then (i) $(r, h)$ must be a blocking pair with respect to $M'$ in $I'$, and (ii) $\{r, r'\} \in E$ for some $r' \in M'(h)$. By construction, for every edge in $E$, one resident is in $R_0$ and the other in $R_1$. If $r \in R_1$ then $r$ cannot form any blocking pair with respect to $M'$ as he is matched to his only choice. If $r \in R_0$ then $r' = r_{n_1+j} \in R_1$ for some $j$ ($1 \leq j \leq n_2$), and $h = M'(r_{n_1+j}) = h_j$. Thus $(r, h) \in A$. By the construction of the preference lists in $(I', G')$, as $(r, h)$ is a (classical) blocking pair of $M'$ in $I'$, $(r, h)$ is also a (classical) blocking pair of $M$ in $I$. Hence $(r, h)$ socially blocks $M$ in $(I, G)$, a contradiction.  

Conversely suppose $M'$ is a resident-complete locally stable matching in $(I', G')$. Then there is no blocking pair $(r, h)$ of $M'$ in $I'$ such that $\{r, r'\} \in E$ for some $r' \in R$ where $r' \in M'(h)$. Let $M = M' \backslash \{(r_{n_1+j}, h_j) : r_{n_1+j} \in R_1\}$. Clearly $M$ is a resident-complete matching in $I$. If some pair $(r, h)$ socially blocks $M$ in $(I, G)$, (i) $(r, h)$ must block $M$ in $I$ (and thus $M'$ in $I'$) and (ii) $(r, h) \in A$. By construction, if $(r, h) \in A$, then $\{r, r'\} \in E$ where $r \in R_0$,  $r' = r_{n_1+j} \in R_1$ and $h  = h_j$. But as $M'$ is resident-complete, $(r_{n_1+j}, h_j) \in M'$ for each $j~(1 \leq j \leq n_2)$.  Thus $r' = r_{n_1+j} \in M'(h)$, a contradiction to the initial assumption that $M'$ is locally stable in $(I', G')$. \qed
\end{proof}

Although the converse statement in Theorem \ref{hrsn_reduction} places a severe restriction on $M'$ (it must be a resident-complete locally stable matching in the HR+SN instance $(I',G')$), it can be relaxed slightly to the case that $M'$ is any locally stable matching in which all the residents $r_{n_1+j} \in R_1$ are matched. It remains to be shown that a reduction exists from HRSS to HR+SN that does not place such a restriction on $M'$.

\section{Hardness of HRSS}
\label{hrf}
We now show that MAX SMISS, the problem of finding a maximum socially stable matching given an SMISS instance is NP-hard. Indeed we prove NP-completeness for COM SMISS, the problem of deciding whether there exists a complete socially stable matching (i.e., a socially stable matching in which all men and women are matched) in an instance of SMISS. It is obvious that, given a matching $M$ of size $k$, it can be verified in polynomial time whether $M$ is socially stable, thus the problem is in NP. Next we define the \emph{Stable Marriage problem with Ties and Incomplete Lists (SMTI)} as a variant of the classical Stable Marriage problem \cite{GS62} where men and women are allowed to find each other unacceptable (thus causing incomplete preference lists) and preference lists may include ties, representing indifference. For example two women may be tied in a given man's preference list.  It was shown in \cite{MIIMM02} that COM SMTI, the problem of deciding whether a complete stable matching exists in an instance of SMTI, is NP-complete even if the ties occur in the men's lists only and each tie occurs at the tail of some list. An SMTI instance $I$ satisfying these restrictions can be reduced to an  SMISS instance $(I', G)$ in polynomial time such that a matching $M$ is a complete stable matching in $I$ if and only if $M$ is a complete socially stable matching in $(I', G)$. These observations lead to the following result. 

\begin{theorem1}
\label{com-smti}
COM SMISS is NP-complete.
\end{theorem1}

\begin{proof}
Consider an instance $I$ of SMTI where the ties occur only on the mens' preference lists and each man has one tie which occurs at the end of the list (a tie may be of length $1$ for this purpose). We define $t(m_i)$ as the set of women contained in the tie in man $m_i$'s preference list. We can construct an instance $(I', G)$ of SMISS  such that $I'$ is the Stable Marriage instance with incomplete lists formed by breaking the ties in $I$ in an arbitrary manner. Let $G=(\mathcal{U} \cup \mathcal{W}, A)$, where $\mathcal{U}$ and $\mathcal{W}$ are the sets of all men and women in $I$ respectively and $A = \bigcup_{m_i \in \mathcal{U}} \bigcup_{w_j \notin t(m_i)} (m_i, w_j)$.  We claim that a matching $M$ is a complete  stable matching in $I$ if and only if $M$ is a complete socially stable matching in $(I', G)$.

Suppose $M$ is a complete stable matching in $I$. Suppose also that $M$ is not socially stable in $(I', G)$. Then there exists some pair $(m_i, w_j) \in A$ that socially blocks $M$ in $(I', G)$. Since $(m_i, w_j) \in A$, $w_j \notin t(m_i)$. Thus $m_i$ prefers $w_j$ to $M(m_i)$ in $I$. Also $w_j$ prefers $m_i$ to $M(w_j)$ since there are no ties in $w_j$'s preference list. Thus $(m_i, w_j)$ blocks $M$ in $I$, a contradiction to our initial assumption.

Conversely, suppose $M$ is a complete socially stable matching in $(I', G)$. Suppose also that $M$ is not stable in $I$. Then there exists some pair $(m_i, w_j)$ that blocks $M$ in $I$. If $M(m_i) \in t(m_i)$ then $w_j \notin t(m_i)$ so $m_i$ prefers $w_j$ to $M(m_i)$ in $(I', G)$. If $M(m_i) \notin t(m_i)$ then $m_i$ prefers $w_j$ to $M(m_i)$ in $(I', G)$. $w_j$ has the same preference list in $I$ and $(I', G)$. So $(m_i, w_j) \in A$ and thus the pair socially blocks $M$ in $(I', G)$, a contradiction to our initial assumption. \qed
\end{proof} 

As discussed in \cite{IMO09}, some centralised matching schemes usually require the agents in one or more sets to have preference lists bounded in length by some small integer. For example, until recently, in the Scottish Foundation Allocation Scheme (the centralised clearinghouse for matching medical residents in Scotland) \cite{Irv98}, medical graduates were required to rank only 6 hospitals in their preference lists. We denote by $(p, q)$-MAX HRSS the problem of finding a maximum socially stable matching in an HRSS instance where each resident is allowed to rank at most $p$ hospitals and each hospital at most $q$ residents. We set $p=\infty$ and $q=\infty$ to represent instances where the residents and hospitals respectively are allowed to have unbounded-length preference lists. Analogously we may obtain the definition of $(p, q)$-MAX SMISS and $(p, q)$-COM SMISS from MAX SMISS and COM SMISS respectively. It turns out that $(p, q)$-COM SMISS is NP-complete even for small values of $p$ and $q$.

\begin{theorem1}
\label{3-3com-smti}
(3,3)-COM SMISS is NP-complete.
\end{theorem1}

\begin{proof}
We prove this by inspecting the hardness result described for the $(3,3)$-COM SMTI problem by Irving et al. \cite{IMO09}. They showed that $(3,3)$-COM SMTI, the problem of deciding whether a complete stable matching exists in an instance of SMTI where each preference list is of length at most $3$, is NP-complete using a reduction from a variant of the SAT problem. 

By inspecting the instance $I$ of SMTI, constructed in the proof, we observe that all the ties appear on the women's side of the instance and appear at the ends of the preference list. We have shown in Theorem \ref{com-smti} that an instance $I$ of SMTI in this form can be reduced in polynomial time to an instance $(I', G)$ of SMISS such that a matching $M$ is a complete stable matching in $I$ if and only if $M$ is a complete socially stable matching in $(I', G)$. We conclude that $(3,3)$-COM SMISS is also NP-complete.
\qed
\end{proof}

\section{Approximating MAX HRSS}
\label{approx_smif}
As shown in Section \ref{hrf}, MAX HRSS is NP-hard. In order to deal with this hardness, polynomial-time approximation algorithms can be developed for MAX HRSS. In this section we present a $3/2$-approximation algorithm for MAX HRSS.  We show this is tight assuming the Unique Games Conjecture (UCG), and also show a $21/19 - \varepsilon$ lower bound assuming $P\not=NP$. The lower bounds hold even for MAX SMISS.

\subsection{Approximation}
\label{sec_approx}
For the upper bound for MAX HRSS, we observe that a technique known as \emph{cloning} has been described in literature \cite{GI89,RS90}, which may be used to convert an HR instance $I$ into an instance $I'$ of the Stable Marriage problem with Incomplete lists in polynomial time, such that there is a one-to-one correspondence between the set of stable matchings in $I$ and $I'$. A similar technique can be used to convert an HRSS instance to an SMISS instance in polynomial time. 

Let $(I, G)$ be an instance of HRSS where $R=\{r_1, r_2,...,r_{n_1}\}$ is the set of residents and $H=\{h_1, h_2,..., h_{n_2}\}$ is the set of hospitals. Let $c_j$ be the capacity of hospital $h_j \in H$.  We can construct an instance $(I', G')$ of SMISS as follows. Each resident in $(I, G)$ corresponds to a man in $(I', G')$. Each hospital $h_j \in H$ gives rise to $c_j$ women in $(I', G')$ denoted by $h_{j,1}, h_{j,2},...,h_{j,c_j}$ each of whom has the same preference list as $h_j$ in $(I', G')$ but with a capacity of $1$. Each man $r_i \in R$ starts off with the same preference list in $(I', G')$ as he has in $(I, G)$. We then replace each entry on his list by the $c_j$ women $h_{j,1}, h_{j,2},...,h_{j,c_j}$ listed in strict order (increasing on second subscript).  $G'$ has vertex set $R\cup H'$, where $H'=\{h_{j,k} : h_j\in H\wedge 1\leq k\leq c_j\}$, and edge set $A'=\{(r_i,h_{j,k}) : (r_i,h_j)\in A\wedge 1\leq k\leq c_j\}$, with $A=E(G)$ denoting the set of acquainted pairs in $(I,G)$.

\begin{theorem1}
Given an instance $(I, G)$ of HRSS, we may construct in $O(n_1 + c_{max}m)$ time an instance $(I', G')$ of SMISS such that a socially stable matching $M$ in $(I, G)$ can be transformed in $O(c_{max}m)$ time to a socially stable matching $M'$ in $(I', G')$ with $|M'| = |M|$ and conversely, where $n_1$ is the number of residents, $c_{max}$ is the maximum hospital capacity and $m$ is the number of acceptable resident-hospital pairs in $I$. 
\label{cloning}
\end{theorem1}

\begin{proof}
Let $(I, G)$ be an instance of HRSS and $(I', G')$ be an instance of SMISS cloned from $(I, G)$. Let $M$ be a socially stable matching in $(I, G)$. We form a matching $M'$ in $(I', G')$ as follows. For each $h_j \in H$, let $r_{j,1}, r_{j,2},...r_{j,x_j}$ be the set of residents assigned to $h_j$ in $M$ where $x_j \leq c_j$, and $k < l$ implies that $h_j$ prefers $r_{j,k}$ to $r_{j,l}$.  Add $(r_{j,k}, h_j^k)$ to $M'$ $(1 \leq k \leq x_j)$. Clearly $M'$ is a matching in $(I', G')$ such that $|M'| = |M|$, and it is straightforward to verify that $M'$ is socially stable in $(I', G')$.

Conversely let $M'$ be a socially stable matching in $(I', G')$. We form a matching $M$ in $(I, G)$ as follows. For each $(r_i, h_j^k) \in M'$, add $(r_i, h_j)$ to $M$. Clearly $M$ is a socially stable matching in $(I, G)$ such that $|M| = |M'|$.

The complexities stated arise from the fact that $I'$ has $O(n_1 + C)$ agents and $O(c_{max}m)$ acceptable man-woman pairs, where $C$ is the total capacity of the hospitals in $I$.
\qed
\end{proof}

Due to Theorem \ref{cloning}, an approximation algorithm $\alpha$ for MAX SMISS with performance guarantee $c$ (for some constant $c > 0$) can be used to obtain an approximation for MAX HRSS with the same performance guarantee. This can be done by cloning the HRSS instance $(I, G)$ to form an SMISS instance $(I', G')$ using the technique outlined above, applying $\alpha$ to $(I', G')$ to obtain a matching $M'$. This matching can then be transformed to a matching $M$ in $(I, G)$ such that $|M| = |M'|$ (again as in the proof of Theorem \ref{cloning}). Our first upper bound for MAX HRSS is an immediate consequence of the fact that any stable matching is at least half the size of a maximum socially stable matching.

\begin{proposition1}
MAX HRSS is approximable within a factor of 2.
\label{approx_max_2}
\end{proposition1}

\begin{proof}
Let $M$ be a maximum socially stable matching given an instance $(I, G)$ of SMISS and let $M'$ be a stable matching in $I$. Thus $M'$ is a maximal matching in the underlying bipartite graph $G'$ in $I$. Hence $|M'| \geq \beta^+(G')/2$ where $\beta^+(G')$ is the size of a maximum matching in $G'$ \cite{KH78}. Also $\beta^+(G') \geq |M|$ and so $|M'| \geq |M|/2$. \qed
\end{proof}

We now present a $3/2$-approximation algorithm for MAX SMISS. The algorithm relies on the principles outlined in the $3/2$-approximation algorithms for the general case of MAX HRT, the problem of finding a maximum cardinality stable matching given an instance of the Hospitals / Residents problem with Ties, as presented by Kir\'{a}ly \cite{ZK12} and McDermid \cite{McD09}. Given an instance $(I, G)$ of SMISS, the algorithm works by running a modified version of the extended Gale-Shapley algorithm \cite{GS62} where unmatched men are  given a chance to propose again by promoting them on all the preference lists on which they appear.  Let $A$ and $U$ denote the sets of acquainted and unacquainted pairs in $(I,G)$ respectively.

Consider a woman $w_j$ in  $(I, G)$. We denote an \emph{unacquainted man $m_i$ on $w_j$'s preference list} as one where $(m_i, w_j) \in U$. Similarly we denote an \emph{acquainted man $m_i$ on $w_j$'s preference list} as one where $(m_i, w_j) \in A$. For a man $m_i$, we denote $next(m_i)$ as the next woman on $m_i$'s list succeeding the last woman to whom he proposed to or the first woman on $m_i$'s list if he has been newly promoted or is proposing for the first time. During the execution of the algorithm if a man runs out of women to propose to on his list for the first time, he is \emph{promoted}, thus allowing him to propose to the remaining women on his list beginning from the first.  A man can only be promoted once during the execution of the algorithm. If a promoted man still remains unmatched after proposing to all the women on his preference list, he is removed from the instance and will not be part of the final matching.

In the classical Gale-Shapley algorithm \cite{GS62} a woman $w_j$ prefers a man $m_i$ to another $m_k$ if $rank(w_j, m_i) < rank(w_j, m_k)$. We define a modified version of the extended Gale-Shapley algorithm \cite{GI89}, \emph{mod-EXGS}, where a woman does not accept or reject proposals from men solely on the basis of their positions on her preference list, but also on the basis of their status as to whether they are acquainted or unacquainted men on her list and whether they have been promoted.  Given two men $m_i$ and $m_k$ on a  woman $w_j$'s preference list, we define the relations $\triangleleft_{w_j}$, $\triangleleft'_{w_j}$ and $\prec_{w_j}$ in Definition  \ref{def1}
\begin{definition1}
\label{def1}
Let $m_i$ and $m_k$ be any two men on a woman $w_j$'s list. Then
\item 1. $m_i \triangleleft_{w_j} m_k$ if either 
\item ~~~(i) $(m_i, w_j) \in U$, $(m_k, w_j) \in U$, $m_i$ is promoted and $m_k$ is unpromoted or
\item ~~~(ii) $(m_i, w_j) \in A$, $(m_k, w_j) \in U$ and $m_k$ is unpromoted. 
\item 2. $m_i \triangleleft'_{w_j} m_k$ if $m_i~{\mathrlap{/}\triangleleft}_{w_j}~m_k~ , ~m_k~{\mathrlap{/}\triangleleft}_{w_j}~m_i ~and~w_j$ prefers $m_i$ to $m_k$ in the classical sense.
\item We define $\prec_{w_j} = \triangleleft_{w_j} ~\cup ~\triangleleft'_{w_j}$.
\end{definition1}
The relation $\prec_{w_j}$ will be used to determine whether a proposal from a man is accepted or rejected by $w_j$.

The main algorithm \emph{approx-SMISS} (as shown in Algorithm \ref{approx_max_smif}) starts by calling \emph{mod-EXGS} (as shown in Algorithm \ref{approx_max_smif2}) where a \emph{proposal sequence} is started by allowing each man to  propose to women beginning from the first woman on his preference list. If a man $m_i$ proposes to a woman $w_j$ on his list and $w_j$ is matched and $m_i \prec_{w_j} M(w_j)$, then $w_j$ is unmatched from her partner $m_k$ and $m_k$ will be allowed to continue proposing to other women on his list. $w_j$ is then assigned to $m_i$. On the other hand, if   $M(w_j) \prec_{w_j} m_i$ then $w_j$ rejects $m_i$'s proposal. Also if $w_j$ is unmatched when $m_i$ proposes, she is assigned to $m_i$. Irrespective of whether the proposal from $m_i$ is accepted or rejected, if $(m_i, w_j) \in A$ then all pairs $(m_k, w_j)$ such that $rank(w_j, m_k) > rank(w_j, m_i)$ are deleted from the instance. However if $(m_i, w_j) \in U$ no such deletions take place. This proposal sequence continues until every man is either matched or has exhausted his preference list. 

After each proposal sequence (where control is returned to the \emph{approx-SMISS} algorithm), if a promoted man still remains unmatched after proposing to all the women on his preference list, he is removed from the instance. Also if a previously unpromoted man exhausts his preference lists and is still unmatched, he is promoted and a new proposal sequence initiated (by calling \emph{mod-EXGS}).  The algorithm terminates when each man either (i) is assigned a partner,  (ii) has no woman on his preference list or (iii) has been promoted and has proposed to all the women on his preference list for a second time. 

\begin{algorithm}[t]
\small
\caption{approx-SMISS}
\label{approx_max_smif}
\begin{algorithmic}[1]
\STATE initial matching $M = \emptyset;$
\WHILE {some unmatched man  with a non-empty preference list exists}
	\STATE call mod-EXGS;
	\FOR {all $m_i$ such that $m_i$ is unmatched and promoted}
			\STATE remove $m_i$ from instance;
	\ENDFOR
	\FOR {all $m_i$ such that $m_i$ is unmatched, unpromoted and has a non-empty preference list}
			\STATE promote $m_i$;
	\ENDFOR	
\ENDWHILE
\STATE return the resulting matching $M$;
\end{algorithmic}
\normalsize
\end{algorithm}

\begin{algorithm}[t]
\small
\caption{mod-EXGS}
\label{approx_max_smif2}
\begin{algorithmic}[1]
	\WHILE {some man $m_i$ is unmatched and still has a woman left on his list}
		\STATE $w_j = next(m_i)$;
				\IF {$w_j$ is matched in $M$ and $m_i \prec_{w_j} M(w_j)$}
					\STATE $M = M ~\backslash ~ \{(M(w_j), w_j)\}$;
				\ENDIF	

				\IF{$w_j$ is unmatched in $M$}
					\STATE $M = M \cup \{(m_i$, $w_j$)\};
				\ENDIF

				\IF {$(m_i, w_j) \in A$}
					\FOR {each $m_k$ such that $(m_k, w_j) \in \mathcal{A}$ and $rank(w_j, m_k) ~>~ rank(w_j, m_i)$}
						\STATE delete $(m_k, w_j)$ from instance;
					\ENDFOR
				\ENDIF
	\ENDWHILE
\end{algorithmic}
\normalsize
\end{algorithm}

\begin{lemma1}
If Algorithm \ref{approx_max_smif} is executed on an SMISS instance $(I, G)$, it terminates with a socially stable matching $M$ in $(I, G)$.
\label{l6}
\end{lemma1}

\begin{proof}
Suppose $M$ is not a socially stable matching and some pair $(m_i, w_j)$ socially blocks $M$ in $(I, G)$. Hence $(m_i, w_j) \in A$. If $w_j$ is unmatched in $M$ then she never received a proposal from $m_i$ (as if she did, she will never become unmatched afterwards). This implies that $m_i$ must prefer his partner in $M$ to $w_j$ as he never proposed to $w_j$. Thus $(m_i, w_j)$ cannot socially block $M$ in this case. 

On the other hand, suppose $w_j$ is matched in $M$ but prefers $m_i$ to $M(w_j)=m_k$.  Also suppose $m_i$ is either unmatched in $M$ or prefers $w_j$ to $M(m_i)$. Then $m_i$ proposed to $w_j$ during the algorithm's execution or $(m_i, w_j)$ was deleted. In either case, all successors of $m_i$ on $w_j$'s list will be deleted, so $(m_k, w_j) \notin M$, a contradiction \qed
\end{proof}

\begin{lemma1}
\label{l7}
During any execution of the algorithm \emph{mod-EXGS}, if $m_i$ proposes to $w_j$ and $(m_i, w_j) \in A$ then $w_j$ will never reject $m_i$ if $rank(w_j, m_i) < rank(w_j, M(w_j))$.
\end{lemma1}

\begin{proof}
This follows from our definition of the $\prec_{w_j}$ relation. Suppose that $w_j$ rejects $m_i$ for a some man $m_k$ and $rank(w_j, m_i) < rank(w_j, m_k)$. Thus $m_k~\prec_{w_j}~m_i$. This implies that  $m_k~\triangleleft_{w}~m_i$ or $m_k~\triangleleft'_{w}~m_i$. Since $(m_i, w_j) \in A$ then $m_k~{\mathrlap{/}\triangleleft}_{w}~m_i$ so $m_k~\triangleleft'_{w}~m_i$ which in turn implies that $rank(w_j, m_k) < rank(w_j, m_i)$, a contradiction to our assumption.
 \qed
\end{proof}

The execution of  the \emph{mod-EXGS} algorithm takes $O(m)$ time where $m=|\mathcal{A}|$ is the number of acceptable pairs. These executions can be performed at most $2n_1$ times, where $n_1$ is the number of men, as a man is given at most two chances to propose to the women on his list. Thus the overall time complexity of the algorithm is $O(n_1m)$. The above results, together with Theorem \ref{cloning}, lead us to state the following theorem concerning the performance guarantee of the approximation algorithm for MAX HRSS.

\begin{theorem1}
MAX HRSS is approximable within a factor of 3/2.
\label{main_approx_theorem}
\end{theorem1}

\begin{proof}
We prove this result by  adopting techniques similar to those used by Iwama et al. \cite{IMY07} and subsequently by Kir\'{a}ly \cite{ZK12}.  We consider alternating paths of odd-length in connected components of the union $M \cup M_{opt}$. It is easy to see that, for alternating paths of length greater than 3, the number of edges in $M_{opt}$ is at most 3/2  times the number of edges in $M$. We now show that alternating paths of length 3 cannot exist in $M \cup M_{opt}$. 

Consider an alternating path of length 3 $\langle(m,w'), (m,w), (m',w)\rangle $ such that $(m,w') \in M_{opt}$, $(m,w) \in M$ and $(m',w) \in M_{opt}$.  Since $w'$ is unmatched in $M$, she was never proposed to during the entire execution of the algorithm. So $w'$ did not delete any men from her preference list and $m$ is unpromoted (if he had been promoted, he would have proposed to $w'$). Thus $m$ prefers $w$ to $w'$. Also since $m'$ is unmatched in $M$, either (i) the pair $(m', w)$ was deleted from the instance at some point during the execution of the algorithm or (ii) $w$ rejected $m'$ twice during the execution of the algorithm.

Consider case (i): if $(m', w)$ was deleted during the execution of the algorithm then $w$ received a proposal from some man $m''$ such that $(m'', w) \in A$ and $w$ prefers $m''$ to $m'$. Thus all successors of $m''$ on $w$'s list would also have been deleted and so $w$ prefers $m$ to $m''$ and hence to $m'$. Although $w$ would have accepted the proposal from $m''$ temporarily (see Lemma \ref{l7}), she ended up with an unpromoted man $m$. It might be that $m=m''$. On the other hand, suppose she accepted a series of proposals from men after $m''$ proposed to $w$ $\langle m''_0, m''_1, ..., m''_k\rangle$ for some $k \geq 0$ before finally being assigned to $m$.  Thus $m''_0 \prec_{w} m''$, which implies that $m''_0~\triangleleft_{w}~m''$ or $m''_0~\triangleleft'_{w}~m''$. Since $(m'', w) \in A$, $m''_0~{\mathrlap{/}\triangleleft}_{w}~m''$ which means $m''_0 \triangleleft'_{w} m''$. But for that to be true, $m''~{\mathrlap{/}\triangleleft}_{w}~m''_0$ as well. For that to happen, $m''_0$ must not be an unacquainted unpromoted man on $w$'s list. This means $m''_0$ can be an acquainted man or an unacquainted promoted man on $w$'s list. If $m=m''_0$ (i.e. $k=0$) then we know that $m$ is unpromoted and so $(m, w) \in A$. On the other hand if another man $m_1''$ exists in the sequence, then the same argument follows that $m''_1 \prec_{w} m''_0$ means that $m''_1 \triangleleft_{w} m''_0$ or $m''_1 \triangleleft'_{w} m''_0$. As already observed, $m''_0$ is either an acquainted man or an unacquainted, promoted man on $w$'s list. In both cases, $m''_1~{\mathrlap{/}\triangleleft}_{w}~m''_0$ meaning $m''_1 \triangleleft'_{w} m''_0$. This implies that $m''_0~{\mathrlap{/}\triangleleft}_{w}~m''_1$. But for that condition to be satisfied, $m''_1$ must not be an unacquainted  unpromoted man on $w$'s list. Once again this means $m''_1$ can be an acquainted man or an unacquainted promoted man on $w$'s list. The same sequence can continue for all men in the sequence until $m$ proposes. Since we already established that $m$ is unpromoted, it follows that $(m, w) \in A$. 

Now consider case (ii): $w$ rejected $m'$ even when he was promoted because of a proposal from some man $m''$. Thus $m'' \prec_{w} m'$ means $m''~\triangleleft_{w}~m'$ or $m''~\triangleleft'_{w}~m'$. Since we know that $m'$ was promoted, then $m''~{\mathrlap{/}\triangleleft}_{w}~m'$ thus $m''~\triangleleft'_{w}~m'$. This means that  $m'~{\mathrlap{/}\triangleleft}_{w}~m''$ and $w$ prefers $m''$ to $m'$. Since $m'~{\mathrlap{/}\triangleleft}_{w}~m''$ , $m''$ must not be an unacquainted unpromoted man on $w$'s list as $m'$ is promoted.  Thus $m''$ must be promoted or an acquainted man on $w$'s list. If $m''$ were an acquainted man on $w$'s list, the pair ($m', w$) would be deleted and the logic presented in case (i) above would follow through. Now suppose $m''$ is promoted and unacquainted on $w$'s list. We know that $m \neq m''$ as $m$ was unpromoted in $M$. Then $w$ may have accepted a series of proposals from men $\langle m''_0, m''_1, ..., m''_k\rangle$ for some $k \geq 0$ before finally being assigned to $m$. Thus $m''_0 \prec_{w} m''$ means $m''_0~\triangleleft_{w}~m''$ or $m''_0~\triangleleft'_{w}~m''$. Since $m''$ is promoted and unacquainted on $w$'s list $m''_0~{\mathrlap{/}\triangleleft}_{w}~m''$ meaning $m''_0 \triangleleft'_{w} m''$. This implies that $m''~{\mathrlap{/}\triangleleft}_{w}~m''_0$. But for that condition to be satisfied, $m''_0$ must not be an unacquainted unpromoted man on $w$'s list. A similar argument to the one presented for case (i) above results in the conclusion that $(m, w) \in A$.

The following conditions  (i) $m$ prefers $w$ to $w'$ (ii) $w$ prefers $m$ to $m'$ and (iii) $(m,w) \in A$ imply that $(m, w)$ will socially block $M_{opt}$, a contradiction.
\qed
\end{proof}

The SIMSS instance shown in Figure \ref{smif_tight} shows that the $3/2$ bound for the algorithm is tight. Here $M_{opt} = \{ (m_1, w_3), (m_2, w_1), (m_3, w_2)\}$ is the unique maximum socially stable matching. Also the approximation algorithm outputs $M = \{(m_1, w_1), (m_2, w_2) \}$  irrespective of the order in which proposals are made. Clearly this instance can be replicated to obtain an arbitrarily large SMISS instance for which the performance guarantee is tight.

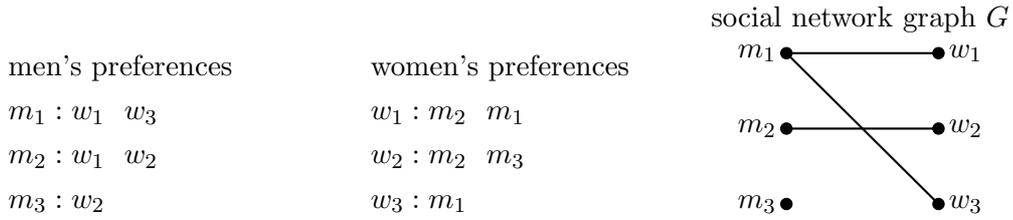
\begin{figure}[t]
\centering
\begin{minipage}[c]{0.6\linewidth}
\begin{align*}
~~&\mbox{men's preferences} ~~~~~~~~~~ &&\mbox{women's preferences} \\
~~&m_1: w_1~~w_3  ~~~~~ &&w_1: m_2~~m_1 \\
~~&m_2: w_1~~w_2 &&w_2: m_2~~m_3 \\
~~&m_3: w_2 &&w_3: m_1
\end{align*}
\end{minipage}
\begin{minipage}[c]{0.3\linewidth}
\centering
social network graph $G$\\
\begin{tikzpicture}
\draw [thick] (0,2) --(2, 2);
\draw [thick] (0,2) --(2,0);
\draw [thick] (0,1) --(2,1);
\draw [fill] (0,0) circle [radius=0.075];
\draw [fill] (0,1) circle [radius=0.075];
\draw [fill] (2,0) circle [radius=0.075];
\draw [fill] (2,1) circle [radius=0.075];
\draw [fill] (0,2) circle [radius=0.075];
\draw [fill] (2,2) circle [radius=0.075];
\node [left] at (0,2) {$m_1$};
\node [left] at (0,1) {$m_2$};
\node [left] at (0,0) {$m_3$};
\node [right] at (2,2) {$w_1$};
\node [right] at (2,1) {$w_2$};
\node [right] at (2,0) {$w_3$};
\end{tikzpicture}
\end{minipage}
\caption{$|M_{opt}| = (3/2).|M|$}
\label{smif_tight}
\end{figure}
We remark that a similar 3/2-approximation algorithm for MAX HRSS was presented independently by Askalidis et al.\ in \cite{AIP13}.

\subsection{Inapproximability}
\label{sec_inapprox}
We now show lower bounds on the approximability of our problem.  We start by giving the inapproximability result assuming P$\not=$NP.

\begin{theorem1}
MAX SMISS is not approximable within $21/19 - \varepsilon$ for any $\varepsilon > 0$, unless P=NP. 
\label{approx_max_21_19}
\end{theorem1}

\begin{proof}
We rely on a proof of the NP-hardness of approximating MAX SMTI given in \cite{HIMY07}. It is shown that there is no approximation algorithm for MAX SMTI with performance guarantee of  $21/19 - \varepsilon$ for any $\varepsilon > 0$, unless P=NP.  This is shown to be true even for instances where ties appear on one side only, each preference list is either strictly ordered or has a single tie of length 2 and ties appear at the end of the preference lists. Thus the same reduction shown in the proof of Theorem \ref{com-smti} can be used to construct an SMISS instance $(I, G)$ such that a polynomial-time algorithm that approximates MAX SMISS to within a factor of $21/19 - \varepsilon$ would do the same for MAX SMTI. \qed
\end{proof}


\par We can get a better lower bound of $3/2 -\varepsilon$, for any $\varepsilon>0$, if we strengthen our assumption from $P\neq NP$ to the truth of the UCG. We do that by showing a reduction from the NP-hard problem of Independent Set (IND SET) and then using an inapproximability result for it from \cite{austrin2009inapproximability} . Along with the hardness of approximation, the reduction gives us a second proof of the NP-hardness of MAX HRSS

\par An instance of IND SET consists of a pair $(G,k)$ of a graph and an integer $k\geq 0$ and asks if the graph $G$ contains a set $S\subseteq V(G)$ such that $|S|=k$ and for every $v_1, v_2\in S, (v_1,v_2)\notin E(G)$, i.e. no two nodes of $S$ are connected with an edge. It is well-known that this problem is NP-hard.

\par Our reduction follows very closely the reduction by Hoefer and Wagner~\cite{HW13} from IND SET to the problem of calculating the maximum locally stable matching. Although the basic idea of the construction is the same, we need to adjust the details as well as the proofs to our problem.

\par Given a graph $G=(V, E)$ we create an instance $(I,G'=(\mathcal{U}\cup \mathcal{W}, A))$ of the $SMISS$ matching problem as follows: First we enumerate (arbitrarily but consistently) all the nodes of $G: v^1, v^2,\ldots, v^n$, where $n=|V(G)|$. For every $v^i\in V(G)$ we create four vertices of $G'$: $m_{v_1^i}, m_{v_2^i}\in \mathcal{U}$ and $w_{v_1^i}, w_{v_2^i}\in \mathcal{W}$. For every $v^i\in V(G)$ we denote by $N_G(v^i)=\{v^j \mid (v^i,v^j)\in E(G)\}$ the neighborhood of $v^i$ in $G$, i.e. the nodes that $v^i$ is connected with in $G$.

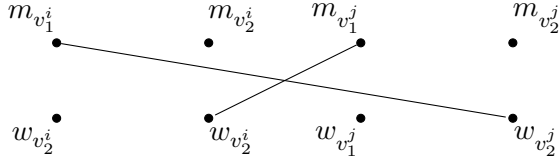
\begin{figure}[t]\label{example}
\setlength{\unitlength}{1.2cm}
    \begin{center}
\begin{tikzpicture}[>=stealth',shorten >=1pt,node distance=2cm,on grid,initial/.style    ={}]
\tikzstyle{every node}=[draw,circle,fill=black,minimum size=3pt,
                            inner sep=0.3pt]


\draw (0,2) node (m11) [label=above left:$m_{v^i_1}$]{};
\draw (2,2) node (w11) [label=above right:$m_{v^i_2}$]{};

\draw (0,1) node (m12) [label=below left:$w_{v^i_2}$]{};
\draw (2,1) node (w12) [label=below right:$w_{v^i_2}$]{};

\draw (4,2) node (m21) [label=above left:$m_{v^j_1}$]{};
\draw (6,2) node (w21) [label=above right:$m_{v^j_2}$]{};

\draw (4,1) node (m22) [label=below left:$w_{v^j_1}$]{};
\draw (6,1) node (w22) [label=below right:$w_{v^j_2}$]{};

    \tikzset{every node/.style={fill=white}} 

  \path (m11)  edge [-,]   (w22);
  \path (m21)  edge [-,]   (w12);
\end{tikzpicture}
\end{center}
\caption{Reduction from IND SET: The gadget of some edge ($v^i,v^j)$}
\end{figure}

For every $v^i\in V(G)$ and every $v^j\in N_G(v^i)$ we add the edge $\{m_{v_1^i}, w_{v_2^j}\}$ to $A(G')$. See Figure 2 for an example of the transformation for an edge $(v^i,v^j)$. The preferences of the agents are as follows (listed in order of preference with most prefered match first):

\begin{quote}
Man $m_{v_1^i}: w_{v_2^i}, \{w_{v_2^j}$ in increasing order on $j \mid v^j\in N_G(v^i)\}, w_{v_1^i}$\\
Man $m_{v_2^i}: w_{v_2^i}$\\
Woman $w_{v_1^i}: m_{v_1^i}$\\
Woman $w_{v_2^i}: m_{v_1^i}, \{m_{v_1^j}$ in increasing order on $j \mid v^j\in N_G(v^i)\}, m_{v_2^i}$
\end{quote}

\noindent
Notice that the enumeration of the vertices of $V(G)$ that we did in the beginning of our reduction shows up here, in the preferences of $m_{v_1^i}$ and $w_{v_2^i}$.   

The idea behind our construction is as follows.  Each man is either of type 1 ($m_{v^i_1}$) or of type 2 ($m_{v^i_2}$), and similarly each woman is of type 1 or type 2.  Type 1 women and type 2 men are both quite picky: $w_{v^i_1}$'s only acceptable partner is $m_{v^i_1}$ and similarly $m_{v^i_2}$'s only acceptable partner is $w_{v^i_2}$, large matchings ought to include $(m_{v^i_1},w_{v^i_1})$ and $(m_{v^i_2},w_{v^i_2})$.  However, the other partners in these pairs, $m_{v^i_1}$ and $w_{v^i_2}$, like this situation least among all acceptable options.  In fact, type 1 men like their own type 2 woman best, and {\it any} type 2 woman from a neighboring node better than their own type 1 woman (and analogously for type 2 women).  Furthermore, through our construction, type 1 men and type 2 women from neighboring nodes $(v^i,v^j)\in E(G)$ are connected in the social graph.  Thus the ``large matching'' for agents corresponding to a node $v^i$ and a node $v^j$ (i.e., $\{(m_{v^i_1},w_{v^i_1}),(m_{v^i_2},w_{v^i_2}),(m_{v^j_1},w_{v^j_1}),(m_{v^j_2},w_{v^j_2})\}$) can not be simultaneously supported by a socially stable matching.  

We prove the following two lemmas for the above construction.
\begin{lemma1}\label{istosocstable}
If $G$ has an independent set $S$ of size $r$ then $I$ has a socially stable matching $M$ of size $n+r$.
\end{lemma1}
\begin{proof}
We take $M=\{(m_{v_1^i}, w_{v_2^i})\mid v^i\in V(G)\setminus S\}\cup \{(m_{v_1^i}, w_{v_1^i}), (m_{v_2^i}, w_{v_2^i})\mid v^i\in S\}$. We obviously have that $|M|=n-|S|+2\cdot |S|=n+r$.
The pairs $(m_{v_1^i}, w_{v_2^i})$ are always socially stable. For the rest of the pairs, the independent set property guarantees that for every $v^j\in N_G(S)$, the matching of $(m_{v_1^j}, w_{v_2^j})$ keeps $(m_{v_1^i}, w_{v_1^i})$ and $(m_{v_2^i}, w_{v_2^i})$ socially stable.  \qed
\end{proof}

The reverse direction is a little more technical.
\begin{lemma1}\label{socstabletois}
If $I$ has a socially stable matching of size $n+r$ then $G$ has an independent set of size $r$.
\end{lemma1}
\begin{proof}
First of all we can assume that $M$ is such that all $m_{v_1^i}$ are matched. If there is an $m_{v_1^i}$ that is not matched we can just take $w_{v_2^i}$ and match her to him. If some other $m_{v_1^j}$ gets single by this action we just do it repeatedly until no such $m_{v_1^i}$ exists. This will not break social stability of the matching nor decrease the size. 
\begin{claim}
If there exists $i\neq j$ such that $(m_{v_1^i}, w_{v_2^j})\in M$, then $(m_{v_1^j}, w_{v_2^i})\in M$ as well.
\end{claim}
\begin{proof}
Notice first, that by our construction and the fact that $(m_{v_1^i}, w_{v_2^j})\in M$ we get that $v^j\in N_G(v^i)$. And hence $(m_{v_1^j}, w_{v_2^i})\in A(G')$.  Suppose now that the statement is not true and take the minimum $i$ such that $(m_{v_1^i}, w_{v_2^j})\in M$ for some $j\neq i$. We first have that $i<j$. For suppose that $j<i$, then if $m_{v_1^j}$ is matched to some $w_{v_2^k}$, for some $k\neq j$, that contradicts the minimal choice of $i$. Hence it must be that $m_{v_1^j}$ is matched to $w_{v_1^j}$. Since $(m_{v_1^j}, w_{v_2^i})\in A(G')$, and $m_{v_1^j}$ prefers $w_{v_2^i}$ over $w_{v_1^j}$,  in order for $M$ to be socially stable it must be that $w_{v_2^i}$ is matched to someone that she prefers to $m_{v_1^j}$ and that can only be some $m_{v_1^k}$ for some $k<j<i$, which contradicts the minimal choice of $i$. So indeed, it must be  that $i<j$.\\
Since $(m_{v_1^j},w_{v_2^i})\notin M$ and they are connected with an edge in $G'$, in order for $M$ to be socially stable it needs to be that at least one of $m_{v_1^j}$ and $w_{v_2^i}$ gets a match under $M$ that he/she prefers to the other.

Suppose that $w_{v_2^i}$ is matched to some $m_{v_1^k}$ for some $k<j$ (these are the only men that $w_{v_2^i}$ prefers to $m_{v_1^j}$, since her top choice, $m_{v_1^i}$, is matched with someone else). Then $(m_{v_1}^i,w_{v_2}^k)\in A(G')$ by construction and so they can not block $M$.  But $m_{v_1^i}$ prefers $w_{v_2^k}$ to $w_{v_2^j}$ (since $k<j$), and $w_{v_2^k}$ prefers $m_{v_1^i}$ to her current match unless her current match is a $m_{v_1^m}$ for some $m\neq k$ and $m<i$. But that would contradict the minimal selection of $i$. So it can't be the case that $w_{v_2^i}$ is matched under $M$ to someone she prefers to $m_{v_1^j}$.

Suppose now that $m_{v_1^j}$ is matched to some $w_{v_2^k}$ for some $k<i$. That would mean that $m_{v_1^k}$ is not matched to $w_{v_2^k}$ and if he is matched to some $w_{v_2^m}$ for some $m\neq k$, that would contradict the minimal choice of $i$. Hence $m_{v_1^k}$ must be matched to $w_{v_1^k}$. Since $m_{v_1^j}$ is matched to $w_{v_2^k}$ and $k\neq j$, by the way we constructed the preferences, that means that $(v^j, v^k)\in E(G)$ and hence $(m_{v_1^k},w_{v_2^j})\in A(G')$. But that contradicts the social stability of $M$ because $w_{v_2^j}$ prefers $m_{v_1^k}$ to $m_{v_2^i}$ (since $k<i$), $m_{v_1^k}$ prefers $w_{v_2^j}$ to $w_{v_1^k}$ (since $w_{v_1^k}$ is his last choice). Hence, it can't be the case that $m_{v_1^j}$ is matched under $M$ to someone he prefers to $w_{v_2^i}$.

So it needs to be that $(m_{v_1^j}, w_{v_2^i})\in M$. \end{proof}

We can now take all dyads of pairs of the form $\{(m_{v_1^i}, w_{v_2^j}), (m_{v_1^j}, w_{v_2^i})\} $ for $i\neq j$ that belong to $M$ and replace them with $\{(m_{v_1^i}, w_{v_2^i}),(m_{v_1^j}, w_{v_2^j})\}$. This will not break social stability of $M$ nor decrease it's size. We take now $S=\{v^i \mid m_{v_2^i}\in M\}$. Since all $n$ of $m_{v_1^i}$ are matched in $M$ there are $r$ couples that contain some $m_{v_2^i}$ and hence $|S|=|M|-n=r$. $S$ is an independent set because $m_{v_2^i}$ can only be matched to $w_{v_2^i}$ and hence $m_{v_1^i}$ is matched to $w_{v_1^i}$. The latter pairs are only stable if for every $v^j\in N_G(v^i)$ it holds that $m_{v_1^j}$ is matched to $w_{v_2^j}$ and hence $m_{v_2^j}\notin M$ which means that $v^j\notin S$ and hence $S$ is indeed an Independent Set.  \qed
\end{proof}

\noindent Lemmas \ref{istosocstable} and \ref{socstabletois} combined give us a second proof of NP-hardness of MAX SMISS. 

Austrin et al. in \cite{austrin2009inapproximability} show that assuming the unique games conjecture IND SET is hard to approximate within a factor of $O\left(\frac{d}{\log ^2d}\right)$ for Independent sets of size $k=\left(\frac{1}{2}-\Theta\left(\frac{\log(\log d)}{\log d}\right)\right)n$, where $n$ is the size of the vertex set and $d$ is an upper bound on the degree. Combining that result and Lemmas \ref{istosocstable} and \ref{socstabletois} we get the following theorem.

\begin{theorem1}
Assuming the UGC, MAX SMISS cannot be approximated within a $3/2-\varepsilon$, for any constant $\varepsilon>0$.
\end{theorem1}
\begin{proof}
The proof follows from the relationship between SMISS and IND SET shown by Lemmas \ref{istosocstable} and \ref{socstabletois} and the result from \cite{austrin2009inapproximability}. It's exactly the same as the proof of Corollary 4 of \cite{HW13}. We set $d=\delta n$ for some constant $\delta>0$ and we get that finding the maximum socially stable matching is hard to approximate within a factor of
\begin{equation}
\frac{n+\left(\frac{1}{2}-\Theta\left(\frac{\log(\log n)}{\log n}\right)\right)n}{n+\left(\frac{1}{2}-\Theta\left(\frac{\log(\log n)}{\log n}\right)\right)n\cdot O\left(\frac{\log n}{n}\right)}\geq \frac{3}{2}-\varepsilon
\end{equation}

for sufficiently large $n$.\qed
\end{proof}

\section{Some special cases of HRSS}
\label{special}
Given the hardness results obtained for the problem of finding a maximum socially stable matching in a general HRSS instance, the need arises to investigate special cases of the problem that are tractable. This section describes some polynomial-time solvability results for three special cases of HRSS. Subsection \ref{2infty_sec} gives a polynomial-time algorithm for finding a maximum socially stable matching given an instance of $(2,\infty)$-MAX SMISS. In Subsection \ref{const_f_sec} we provide a polynomial-time algorithm for MAX HRSS in the case where there is a constant number of unacquainted pairs and in Subsection \ref{sec_const_adm} we also consider the case where the number of acquainted pairs is constant, again providing a polynomial time-algorithm for MAX HRSS in that context.

\subsection{$(2, \infty)$-MAX SMISS}
\label{2infty_sec}
Given an SMISS instance $(I, G)$, where the men are allowed to have at most two women in their preference lists and each woman is allowed to have an unbounded-length preference list, we show that a maximum socially stable matching can be found in polynomial time. We make slight modifications to the algorithm used to find a maximum  stable matching in a $(2, \infty)$-MAX SMTI (the problem of finding a maximum stable matching given an SMTI instance where men are allowed to have at most two women on their preference list) instance described in \cite{IMO09}. The resulting algorithm, which we call \emph{$(2,\infty)$-MAX-SMISS-alg}, is broken down into three phases. 

In Phase 1, some pairs that cannot be involved in any socially stable matching in $(I, G)$ are \emph{deleted} from the instance. A pair $(m_i, w_j)$ is deleted by removing $m_i$ from $w_j$'s preference list and vice versa. We call the resulting preference lists the \emph{reduced preference lists}. For each man $m_i$, if the first woman $w_j$ on $m_i$'s preference list satisfies $(m_i, w_j) \in A$, where $A=E(G)$ is the set of acquainted pairs in $(I,G)$, we delete all pairs $(m_k, w_j)$ for all successors $m_k$ of $m_i$ on $w_j$'s preference list. $(m_k, w_j)$ cannot be involved in any socially stable matching as $(m_i, w_j)$ will socially block any matching they were involved in. 

In Phase 2, we construct a weighted bipartite graph $G'$ from the resulting instance with a reduced preference list. This is done by representing the men and women as nodes on the two sides of the graph and adding an edge between a man $m_i$ on one side and a woman $w_j$ on the other if $w_j$ appears on $m_i$'s preference list. The weight placed on the edge will be the position of $m_i$ on $w_j$'s preference list (denoted by $rank(w_j, m_i)$). Algorithm \ref{buildG} describes the process. A minimum cost maximum cardinality matching $M_{G'}$ in the the resulting bipartite graph is then generated using the algorithm described in \cite{GT89}.

At this stage it is not guaranteed that the resulting maximum matching $M_{G'}$ is socially stable. $M_{G'}$ may admit a social blocking pair $(m_i,w_j)$, where $(m_i,w_j)\in A$, $w_j$ is the first-choice partner of $m_i$, $m_i$ is assigned to his second-choice partner $w_k$ in $M_{G'}$ and $w_j$ is unassigned in $M_{G'}$ (as we will show later, this the only form of social blocking pair that $M_{G'}$ can admit). To remove such blocking pairs, during Phase 3, we assign $m_i$ to $w_j$ thus making $w_k$ unmatched. At this stage, $w_k$ may herself be the first-choice woman that forms a pair with some other man $m_l$ in the resulting matching such that $(m_l, w_k) \in A$. The process continues until there is no man who is matched to his second-choice woman while forming an edge in $G$ with his unmatched first-choice woman. Algorithm \ref{2infty} shows the entire \emph{$(2,\infty)$-MAX-SMISS-alg} algorithm.

\begin{algorithm}[t]
\small
\caption{$(2,\infty)$-MAX-SMISS-alg}
\label{2infty}
\begin{algorithmic}[1]
\STATE /* Phase 1 */
\WHILE{some man $m_i$ has a first-choice woman $w_j$ such that $(m_i, w_j) \in A$}
	\FOR {each successor $m_k$ of $m_i$ on $w_j$'s list}
		\STATE delete the pair $(m_k, w_j)$;
	\ENDFOR
\ENDWHILE

\STATE /* Phase 2 */
\STATE $G' := $ buildGraph();
\STATE $M_{G'}$ := minimum weight maximum matching in $G'$;

\STATE /* Phase 3 */
\STATE $M := M_{G'}$;
\WHILE {there exists a man $m_i$ who is matched to his second-choice woman $w_k$ \\ and his first-choice $w_j$ is an unmatched woman such that $(m_i, w_j) \in A$}
	\STATE $M := M \backslash \{(m_i, w_k)\}$; 
	\STATE $M := M \cup \{(m_i, w_j)\}$;
\ENDWHILE
\STATE  return $M$;
\end{algorithmic}
\normalsize
\end{algorithm}

\begin{algorithm}[t]
\small
\caption{buildGraph}
\label{buildG}
\begin{algorithmic}[1]
\STATE $V := \mathcal{U} \cup \mathcal{W}$; /* $\mathcal{U}$ and $\mathcal{W}$ are sets of men and women in $I$ */
\STATE $E' := \emptyset$;
\FOR {each man $m_i \in \mathcal{U}$}
	\FOR {each woman $w_j$ on $m_i$'s reduced list}
		\STATE $E' := E' \cup \{(m_i, w_j)\}$;
		\STATE $weight(m_i, w_j) := rank(w_j, m_i)$;
	\ENDFOR
\ENDFOR
\STATE $G' := (V, E')$;
\STATE  return $G'$;
\end{algorithmic}
\normalsize
\end{algorithm}

We now show that at the end of this phase, for $(2, \infty)$-MAX SMISS instances, the matching produced is both socially stable and of maximum size. 

\begin{lemma1}
\emph{$(2,\infty)$-MAX-SMISS-alg} terminates.
\label{2-max-hrf}
\end{lemma1}

\begin{proof}
It is easy to see that Phases $1$ and $2$ terminate. For every iteration of Phase $3$, one man always improves from his second to his first choice while no man obtains a worst partner or becomes unmatched.  Since the total number of possible improvements is finite, it is clear the phase is bound to terminate. \qed
\end{proof}

\begin{lemma1}
Phase 1 of \emph{$(2,\infty)$-MAX-SMISS-alg} never deletes a socially stable pair, which is a man-woman pair that belongs to some socially stable matching in $(I, G)$.
\label{2-max-hrf2}
\end{lemma1}

\begin{proof}
Suppose pair $(m_i, w_j)$ is deleted during some execution of \emph{$(2,\infty)$-MAX-SMISS-alg} and $(m_i, w_j) \in M$ where $M$ is some socially stable matching in $(I, G)$. Then $m_i$ was deleted from $w_j$'s preference list because $w_j$ was the first-choice woman of some man $m_k$ such that $(m_k, w_j) \in A$ and $w_j$ prefers $m_k$ to $m_i$. Therefore $(m_k, w_j)$ will socially block $M$, a contradiction. \qed
\end{proof}

\begin{lemma1}
The matching returned by \emph{$(2,\infty)$-MAX-SMISS-alg} is socially stable in $(I, G)$.
\label{2-max-hrf3}
\end{lemma1}

\begin{proof}
Suppose the matching $M$ produced by \emph{$(2,\infty)$-MAX-SMISS-alg} is not socially stable in $(I, G)$. Then some pair $(m_i, w_j) \in A$ socially blocks $M$ in $I$. For this to occur, one of the following cases must arise.

Case (i): $m_i$ and $w_j$ are unmatched in $M$. Then $m_i$ is unmatched in $M_{G'}$ and either $w_j$ was initially unmatched in $M_{G'}$ or became unmatched due to some operation in Phase 3 of \emph{$(2,\infty)$-MAX-SMISS-alg}. If $w_j$ was initially unmatched in $M_{G'}$ then $M_{G'}$ could have been increased in size by adding $(m_i, w_j)$ thus contradicting the fact that $M_{G'}$ is of maximum cardinality. Suppose $w_j$ became unmatched due to Phase 3. Let $m_{p_1}$ denote $w_j$'s partner in $M_{G'}$. During Phase 3, $m_{p_1}$ must have become assigned to his first-choice woman $w_{q_1}$. Suppose $w_{q_1}$ was unmatched in $M_{G'}$. Then  a larger matching can be obtained by augmenting $M_{G'}$ along the path $(m_i, w_j), (w_j, m_{p_1}), (m_{p_1}, w_{q_1})$ contradicting to the fact that $M_G$ is of maximum cardinality. Thus $w_{q_1}$ must have been assigned in $M_{G'}$ and became unmatched during Phase 3 as well. If $w_{q_1}$ was assigned to $m_{p_2}$ in $M_{G'}$, $m_{p_2}$ must have become assigned to his first-choice woman $w_{q_2}$ during Phase 3. Using a similar argument to that used for $w_{q_1}$, we can argue that $w_{q_2}$ must have been assigned in $M_{G'}$ as well. Thus some man moved from $w_{q_2}$ to his first-choice woman. This process may be continued and at each iteration of Phase 3, some man must strictly improve and no man becomes worse off. Since the possible number of such improvements is finite, there are a finite number of women that can be unmatched in this way in Phase 3. Hence at some point, there exists a man $m_{p_s}$, who switches to his first-choice woman $w_{q_s}$ and $w_{q_s}$ was already unmatched in $M_{G'}$. We can then construct an augmenting path in $G'$ of the form $(m_i, w_j), (w_j, m_{p_1}), (m_{p_1}, w_{q_1}),(w_{q_1}, m_{p_2}), (m_{p_2}, w_{q_2}),...,(m_{p_s}, w_{q_s})$ which contradicts the fact that $M_{G'}$ is of maximum cardinality. 

Case (ii): $m_i$ is unmatched in $M$ and $w_j$ prefers $m_i$ to $m_l = M(w_j)$. Then $m_i$ is unmatched in $M_{G'}$ as well. Suppose that $w_j$ was assigned to $m_l$ in $M_{G'}$. Since $w_j$ prefers $m_i$ to $m_l$, a matching of equal size to $M_{G'}$ but with a lower weight can be obtained by pairing $m_i$ with $w_j$, leaving $m_l$ unmatched, thus contradicting the fact that $M_{G'}$ is a minimum weight maximum cardinality matching. Thus $w_j$ is not assigned to $m_l$ in $M_{G'}$. Then $w_j$ is either unmatched in $M_{G'}$ or is assigned in $M_{G'}$ to some man $m_p$, where $m_p \neq m_l$ and $m_p \neq m_i$. If $w_j$ is unmatched in $M_{G'}$, the fact that $M_{G'}$ is a maximum cardinality matching is contradicted. If $w_j$ is assigned to $m_p$ in $M_{G'}$ then since $w_j$ is no longer assigned to $m_p$ in $M$, $m_p$ must have switched to his first-choice woman $w_q$ during Phase 3. Hence either $w_q$ was unmatched in $M_{G'}$ or $w_q$ became unmatched due to some man being switched to his first-choice woman. Using a similar argument as in case (i), we can construct an augmenting path that contradicts the fact that $M_{G'}$ is of maximum cardinality. 

Case (iii): $m_i$ is assigned to $w_k$ in $M$ and $m_i$ prefers $w_j$ to $w_k$  and $w_j$ is unmatched in $M$. Since $m_i$'s preference list is of length 2, $w_j$ is $m_i$'s first-choice and $(m_i, w_j) \in A$ and so this case satisfies the loop condition in Phase 3 and thus should never arise once Phase 3 has terminated (as it must, by Lemma \ref{2-max-hrf}).

Case (iv): $m_i$ is assigned to $w_k$ in $M$ and $m_i$ prefers $w_j$ to $w_k$ and $w_j$ is assigned to $m_l$ in $M$ and $w_j$ prefers $m_i$ to $m_l$. Once again, since $m_i$'s list is of length 2, then $w_j$ must be $m_i$'s first-choice and $(m_i, w_j) \in A$. Therefore the loop condition of Phase 1 would have ensured that man $m_l$ was deleted from $w_j$'s preference list. \qed
\end{proof}

Since, by Lemma \ref{2-max-hrf2}, Phase 1 never deletes a socially stable pair, a maximum socially stable matching must consist of pairs that belong to the reduced lists. Since $M_{G'}$ is a maximum matching and Phase 3 never reduces the size of the matching, it follows that the matching produced by the algorithm is of maximum cardinality. Finally since by Lemma \ref{2-max-hrf3}, the matching produced is a socially stable matching, it follows that the algorithm produces a maximum socially stable matching  in $(I, G)$. 

The complexity of the algorithm is dominated by Phase 2. The complexity of the algorithm for finding the minimum cost maximum matching in $G' = (V, E')$ is $O(\sqrt{|V|} |E'| \log{|V|})$ \cite{GT89}. Let $n=|V| = n_1 + n_2$ be the total number of men and women. Since the set $\mathcal A$ of acceptable pairs satisfies $|\mathcal{A}| \leq 2n_1 = O(n)$, it follows that $(2, \infty)$-MAX SMISS has a time complexity of $O(n^{3/2}\log{n})$. We have thus proved the following theorem. 

\begin{theorem1}
Given an instance $(I, G)$ of $(2, \infty)$-MAX SMISS, Algorithm $(2,\infty)$-MAX-SMISS-alg generates a maximum socially stable matching in $O(n^{3/2}\log{n})$ time, where $n$ is the total number of residents and hospitals in $(I, G)$.
\end{theorem1}

\subsection{HRSS with a constant number of unacquainted pairs}
\label{const_f_sec}
It is easy to see that in the special case where the set $U$ of unacquainted pairs is exactly the set $\mathcal A$ of acceptable pairs in the underlying HR instance, then the set $A$ of acquainted pairs satisfies $A = \emptyset$ and every matching found is a socially stable matching. Also if the instance contains no unacquainted pairs (i.e., $A = \mathcal{A}$ and $U = \emptyset$), then only stable matchings in the classical sense are socially stable. In both these cases, a maximum socially stable matching can be generated in polynomial time. The case may however arise where the number of unacquainted pairs is constant. In this case, we show that it is also possible to generate a maximum socially stable matching in polynomial time. 

Let $(I, G)$ be an instance of HRSS and let $S \subseteq \mathcal{A}$ be a subset of the acceptable pairs in $I$. We denote $I \backslash S$ as the instance of HR obtained from $I$ by deleting the pairs in $S$ from the preference lists in $I$. The following proposition plays a key role in establishing the correctness of the algorithm. 

\begin{proposition1}
\label{prep_const_free}
Let $(I, G)$ be an instance of HRSS. Let $M$ be a socially stable matching in $(I, G)$. Then there exists a set of  unacquainted pairs $U' \subseteq U$ such that $M$ is stable in $I' = I \backslash U'$. Conversely suppose that $M$ is a stable matching in $I' = I \backslash U'$ for some $U' \subseteq U$. Then $M$ is socially stable in $(I, G)$. 
\end{proposition1}

\begin{proof}
Suppose $M$ is socially stable in $(I, G)$. Let $U' = U \backslash (M \cap U) = U \backslash M$. We claim that $M$ is stable in $I' = I \backslash U'$. Suppose $(r_i, h_j)$ blocks $M$ in $I'$. Then $(r_i, h_j) \notin M$ and edges in $I'$ are those in $A \cup M \cap U$ and $(r_i, h_j) \notin M \cap U$. Thus $(r_i, h_j) \in A$ making $M$ socially unstable in $(I, G)$, a contradiction. 

Conversely, suppose that $M$ is stable in $I' = I \backslash U'$ for some $U'\subseteq U$. We claim that $M$ is socially stable in $(I, G)$. Suppose that $(r_i, h_j)$ socially blocks $M$ in $(I, G)$, then $(r_i, h_j) \in A$ so $(r_i, h_j) \notin U$. Thus $(r_i, h_j)$ is a pair in $I'$ and so $(r_i, h_j)$ blocks $M$ in $I'$, a contradiction.  
\qed
\end{proof}

By considering all subsets  $U' \subseteq U$, forming $I'$, finding a stable matching in each such $I'$ and keeping a record of the maximum stable matching found, we obtain a maximum socially stable matching in $(I, G)$. This discussion leads to the following theorem.

\begin{theorem1}
\label{th_const_free}
Given an instance $(I, G)$ of HRSS where the set $U$ of unacquainted pairs is of constant size, a maximum socially stable matching can be generated in $O(m)$ time, where $m=|\mathcal{A}|$ is the number of acceptable pairs.
\end{theorem1}

\begin{proof}
By considering all subsets  $U' \subseteq U$, forming $I'$, finding stable a matching in each such $I'$ and keeping a record of the maximum stable matching found, we obtain a maximum socially stable matching in $(I, G)$. Let $M$ be the largest stable matching found over all $I' = I \backslash U'$ where $U'\subseteq U$. Then by Proposition \ref{prep_const_free}, $M$ is socially stable in $(I, G)$. Suppose for a contradiction that $M'$ is a socially stable matching in $I$ where $|M'| > |M|$. Then by Proposition \ref{prep_const_free} there exists some $U'' \subseteq U$ such that $M'$ is stable in $I''$. But that contradicts $M$ as the largest stable matching in $I \backslash U'$ taken over all $U' \subseteq U$.

Finding all subsets of $U$ can be done in $O(2^k)$ time where $k = |U|$. Forming $I'$ and finding a stable matching in $I'$ can be done in $O(|\mathcal{A}|)$ time where $\mathcal{A}$ is the set of acceptable pairs in $I$. Hence the overall time complexity of the algorithm is $O(2^k  |\mathcal{A}|)$ which is polynomial with respect to the instance size as $k$ is a constant.
\qed
\end{proof}

\subsection{HRSS with a constant number of acquainted pairs}
\label{sec_const_adm}
We now consider the restriction of HRSS in which the set $A$ of acquainted pairs is of constant size $k$.  Given an instance $(I, G)$ of this problem we show that a maximum socially stable matching can be found in polynomial time. Let $A = \{e_1, e_2, ..., e_k\}$ where $e_i$ represents an acquainted pair $(r_{s_i}, h_{t_i})$ $(1 \leq i \leq k)$. A tree $T$ of depth $k$ is constructed with all nodes at depth $i$ labelled $e_{i+1}$ $(i \geq 0)$. There are left and right branches below $e_i$. Each branch corresponds to a condition placed on $r_{s_i}$ or $h_{t_i}$ with respect to a matching $M$. The left branch below $e_i$ (i.e., a resident condition branch) corresponds to the condition that $r_{s_i}$ is matched in $M$ and prefers his partner to $h_{t_i}$. The right branch below $e_i$ (i.e., a hospital condition branch) corresponds to the condition that $h_{t_i}$ is fully subscribed in $M$ and has a partner no worse than $r_{s_i}$. Satisfying at least one of these conditions ensures that $M$ admits no blocking pair involving $(r_{s_i}, h_{t_i})$. The tree is constructed in this manner with the nodes at depth $k-1$, labelled $e_k$, branching in the same way to dummy leaf nodes $e_{k+1}$ (not representing acquainted pairs).

A path $P$ from the root node $e_1$ to a leaf node $e_{k+1}$ will visit all pairs in $A$ exactly once. Every left branch in $P$ gives a resident condition and every right branch gives a hospital condition. Let $R'$ and $H'$ be the set of residents and hospitals involved in resident and hospital conditions in $P$ respectively. Given a matching $M$, enforcing all the conditions along $P$ can be achieved by first deleting all pairs from the instance $I$ that could potentially violate these conditions. So if some resident condition along $P$ states that a resident $r_{s_i}$ must be matched in $M$ to a hospital he prefers to $h_{t_i}$ then $r_{s_i}$'s preference list is truncated starting with $h_{t_i}$. If some hospital condition states that a hospital $h_{t_i}$ must be fully subscribed in $M$ and must not be matched to a resident worse than $r_{s_i}$ then $h_{t_i}$'s preference list is truncated starting from the resident immediately following $r_{s_i}$.  After performing these truncations based on the conditions along $P$, a new HR instance $I'$ is obtained. 

\begin{proposition1}
\label{prep_const_adm}
If $M$ is a matching in $I'$ that is computed at the leaf node of a path $P$ and all residents in $R'$ are matched in  $M$ and all hospitals in $H'$ are fully subscribed in $M$ then $M$ is a socially stable matching in $(I, G)$. 
\end{proposition1}

\begin{proof}
Suppose some resident-hospital pair socially blocks $M$ in $(I, G)$. Then this pair belongs to $A$ so it corresponds to a node $e_i = (r_{s_i}, h_{t_i})$ in $T$ for some $i$ $(1 \leq i \leq k)$. So in $M$, either (i) $r_{s_i}$ is unmatched or prefers $h_{t_i}$ to $M(r_{s_i})$ and (ii) either $h_{t_i}$ is undersubscribed or prefers $r_{s_i}$ to $M(h_{t_i})$. Suppose in $T$, the left hand branch from $e_i$ was chosen when following a path $P$ from the root to a leaf node. Then $r_{s_i} \in R'$, $r_{s_i}$ is matched and by the truncations carried out, $r_{s_i}$ has a better partner than $h_{t_i}$, a contradiction. Thus the right hand branch from $e_i$ was chosen when following $P$. Then $h_{t_i} \in H'$, $h_{t_i}$ is fully subscribed and by the truncations $h_{t_i}$ has no partner worse than $r_{s_i}$. But $(r_{s_i}, h_{t_i}) \notin M$, so $h_{t_i}$ has a partner better than $r_{s_i}$, a contradiction. \qed
\end{proof}

With $I'$ obtained due the truncations carried out by satisfying conditions along a path $P$ from the root node to a leaf node, we then seek to obtain a matching in which all the residents in $R'$ are matched and hospitals in $H'$ are fully subscribed. For each hospital $h_j$ in $I'$, we define a set of clones of $h_j$, $\{h_{j,1}, h_{j,2},...,h_{j,{c_j}}\}$, corresponding to the number of posts $c_j$ available in $h_j$. Let $H'' = \{h_{j,k}:h_j \in H' \wedge 1 \leq k \leq c_j \}$ denote the set of all clones obtained from hospitals in $H'$. We define a bipartite graph $G'$ where one set  of nodes is represented by the set of residents in $I'$ and the other set of nodes is represented by the hospital clones in $I'$. If $r_i$ finds $h_j$ acceptable in $I$, we add an edge from $r_i$ to $h_{j,q}$ in $G'$ for all $q~(1 \leq q \leq c_j)$. We define a new graph $G''$ containing the same nodes and edges in $G'$ with  weights placed on the edges. We mark all the nodes representing the residents in $R'$ and the hospital clones in $H''$ as red nodes with the remaining nodes uncoloured. We place weights on the edges as follows: (i) an edge between a red node and an uncoloured node is given a weight of $1$; (ii) an edge between two red nodes is given weight of $2$; (iii) an edge between two uncoloured nodes is given a weight of $0$. We then find a maximum weight matching $M'$ in the resulting weighted bipartite graph $G''$. 
Let $wt(M')$ denote the weight of a matching $M'$ in $G''$. Then 

\begin{figure}[!h]

\begin{minipage}{\linewidth}
\centering
$wt(M') = |\{(r_i, h_{j,q}) \in M' : r_i \in R'\}| + |\{(r_i, h_{j,q}) \in M' : h_{j,q} \in H''\}| \leq |R'| + |H''|$.
\end{minipage}

\end{figure}  

Moreover $wt(M') = |R'| + |H''|$ if and only if every agent in $R' \cup H''$ is matched in $M'$. For such a matching $M'$ in $G''$, by construction, $M'$ is also a matching in $G'$ and a maximum cardinality matching $M''$ can be obtained in $G'$ by continuously augmenting $M'$ until no augmenting path can be found. Since any node already matched in $M'$ will remain matched in $M''$ it follows that all the residents and hospital clones in $R' \cup H''$ will be matched in $M''$ and such a matching, by Proposition \ref{prep_const_adm}, will be socially stable in $(I, G)$. If however, $wt(M') < |R'| + |H''|$ then some resident or hospital clone in $R' \cup H''$ remains unmatched in any maximum matching in $G'$ introducing the possibility of a social blocking pair of $M'$ in $(I, G)$. In this case, $P$ is ruled as infeasible and another path is considered, otherwise $P$ is called feasible.

There are $2^k$ paths from the root node to leaf nodes in the tree $T$. The following proposition is important to our result. 

\begin{proposition1}
\label{prep_const_adm2}
There must exist at least one feasible path in $T$.
\end{proposition1}

\begin{proof}
Let $M$ be a stable matching in $I$. Then $M$ is socially stable in $(I, G)$. Consider each node $e_i = (r_{s_i}, h_{t_i})$ in $T$. If $(r_{s_i}, h_{t_i}) \in M$, we branch right at $e_i$. If $(r_{s_i}, h_{t_i}) \notin M$, and $r_{s_i}$ is matched and has a partner better than $h_{t_i}$, we branch left at $e_i$. If $(r_{s_i}, h_{t_i}) \notin M$, and $h_{t_i}$ is fully subscribed and has no partner worse than $r_{s_i}$, we branch right at $e_i$. Any other condition would mean that $(r_{s_i}, h_{t_i})$ would block $M$ in $I$, a contradiction. This process, starting from the root node $e_1$, gives a feasible path $P$ through $T$ to some leaf node $v$ where, for the set of residents $R'$ and hospitals $H'$ involved in $P$, $wt(M) = |R'| + |H'|$. \qed
\end{proof}

To generate a maximum socially stable matching $M$ in an instance $(I, G)$ of HRSS, all $2^k$ paths through $T$ from the root node to leaf nodes are considered with a record kept of the maximum matching $M''$ computed at the leaf node of each feasible path. The desired matching $M$ can then be constructed by letting
 \[M= \{(r_i, h_j) : (r_i, h_{j,k}) \in M'' \mbox{ for some $k$ } (1 \leq k \leq c_j)\}.\]

\begin{proposition1}
\label{prep_const_adm3}
If $M$ is a matching obtained from the process described above, $M$ is a maximum socially stable matching in (I, G).
\end{proposition1}

\begin{proof}
Proposition \ref{prep_const_adm} shows that $M$ is socially stable in $(I, G)$. Suppose $M'$ is a socially stable matching in $(I, G)$ such that $|M'| > |M|$. Construct a feasible path $P$ through $T$ from the root node to a leaf node $v$, branching left or right at each node $e_i = (r_{s_i}, h_{t_i})$ as follows. If $r_{s_i}$ is matched in $M'$ to some hospital better than $h_{t_i}$, branch left. Otherwise as $(r_{s_i}, h_{t_i})$ is not a social blocking pair of $M'$, $h_{t_i}$ is fully subscribed in $M'$ with no assignee worse than $r_{s_i}$, in which case branch right. As before, construct sets $R'$ and $H'$ as follows. For every left branch in $P$ involving a resident $r_{s_i}$, add $r_{s_i}$ to $R'$ and for every right branch in $P$ involving a hospital $h_{t_i}$, add $h_{t_i}$ to $H'$. Matching $M'$ then satisfies the property that every resident $r_{s_i} \in R'$ is matched in $M'$ to a hospital better than $h_{t_i}$ and every hospital $h_{t_i} \in H'$ is fully subscribed with residents no worse than $r_{s_i}$. At the leaf node $v$ of $P$, the algorithm constructs a matching $M''$ which is of maximum cardinality with respect to the restrictions that every resident $r_{s_i} \in R'$ is matched to a hospital better than $h_{t_i}$ and every hospital $h_{t_i} \in H'$ is fully subscribed with residents no worse than $r_{s_i}$. Hence $|M''| \geq |M'|$ and since $|M'| > |M|$, it follows that $M''$ contradicts the choice of $M$ as the largest matching taken over all leaf nodes. \qed
\end{proof}

The above proposition leads to the following main result of this subsection.

\begin{theorem1}
\label{max-matching-fixed-a}
Given an instance $(I, G)$ of HRSS where the set $A$ of acquainted pairs satisfies $|A| = k$ for some constant $k$, a maximum socially stable matching can be generated in $O(c_{max}m\sqrt{n_1+C})$ time where $n_1$ is the number of residents, $m$ is the number of acceptable pairs, $c_{max}$ is the largest capacity of any hospital and $C$ is the total capacity of all the hospitals in the problem instance. 
\end{theorem1}

\begin{proof}
Since $|A|=k$ is constant, the number of leaf nodes (and subsequently the number of paths from the root node to a leaf node) is also constant ($2^k$).  Performing the truncations imposed by the conditions along a path can be done in $O(m)$ time where $m$ is the number of acceptable pairs in $I$. The number of nodes $n'$ and the number of edges $m'$ in $G'$ are given by:
\begin{eqnarray*}
\vspace{-0.5 cm}
\centering
\begin{minipage}{350px}
$m' = \sum\limits_{j=1}^{n_2} |pref(h_j)| c_j \leq c_{max}\sum\limits_{j=1}^{n_2} |pref(h_j)| \leq c_{max}m$\\
$n' = n_1 + \sum\limits_{j=1}^{n_2} c_j = n_1+C$
\end{minipage}
\vspace{-0.5 cm}
\end{eqnarray*}

\noindent
where $n_1$ is the number of residents, $n_2$ is the number of hospitals, $pref(h_j)$ is the set of residents in $h_j$'s preference list, $c_{max}$ is the largest capacity of any hospital and $C$ is the total number of posts in the problem instance. 
Finding a maximum weight matching in $G''$ can be done in $O(m'\sqrt{n'})$ time \cite{DSU12} (since the edge weights have $O(1)$ size). Augmenting such a matching to a maximum cardinality matching in $G'$ can be done in $O(m'\sqrt{n'})$ time \cite{HK73}. Thus the time complexity of the algorithm is $O(c_{max}m\sqrt{n_1+C})$.
\qed
\end{proof}

Following the results in Theorems \ref{th_const_free} and \ref{max-matching-fixed-a}, we conclude this section with the theorem below showing the existence of FPT algorithms for MAX HRSS under two different parameterisations.

\begin{theorem1}
\label{fpt}
MAX HRSS is in FPT with parameter $k$, where either $k=|A|$ or $k=|U|$, and $A$ and $U$ are the sets of acquainted and unacquainted pairs respectively.
\end{theorem1}

\section{Open problems}
\label{open}
The study of the Hospitals / Residents problem under Social Stability is still at an early stage, and some interesting open problems remain. Firstly it is worth investigating the complexity of finding socially stable matchings with additional optimality properties (such as egalitarian / minimum regret socially stable matchings \cite{ILG87,Gus87}). 
It may also be interesting to determine whether the set of socially stable matchings in an HRSS instance admits any structure similar to that present in the case of stable matchings in HR. Also extending the algorithm for  $(2, \infty)$-MAX SMISS to the HRSS case remains open.  The HRSS model described in this paper does not consider the scenario where ties exist in the preference lists of agents. It remains open to investigate this variant of the problem. 

We close by remarking that currently in the HRSS context, if $(r,h)$ is a blocking pair of a matching $M$ where $h$ is undersubscribed and $(r,h)\in U$, then $(r,h)$ is not a social blocking pair of $M$.  However it could be argued that information about undersubscribed hospitals would be in the public domain, and hence the social blocking pair definition should not require that $r$ and $h$ have a social tie in this case.  It would be interesting to investigate algorithmic aspects of this variant of HRSS.

\section*{Acknowledgments}
\label{ack}
The fourth author is supported by grant EP/K010042/1 from the Engineering and Physical Sciences Research Council. We would like to thank Martin Hoefer for making the third and fourth authors aware of \cite{AIP13}; Zolt\'an Kir\'{a}ly for ideas that led to Theorem \ref{hrsn_reduction}, for observing Theorem \ref{fpt} and for other valuable comments; and Rob Irving and an anonymous referee for further valuable suggestions concerning this paper. 
\bibliography{matching_db}

\end{document}